\documentclass[final]{l4dc2024}

\title[Convex neural network synthesis]{Convex neural network synthesis for robustness in the 1-norm}
\usepackage{times}
\usepackage{subcaption}
\usepackage{wrapfig}

\newcommand{\TR}{\color{black}}

\author{%
 \Name{Ross Drummond} \Email{ross.drummond@sheffield.ac.uk}\\
 \addr Dept. of Automatic Control and Systems Engineering,  University of Sheffield,  Sheffield, S1 3JD, UK.\\
 \AND
 \Name{Chris Guiver} \Email{C.Guiver@napier.ac.uk}\\
 \addr {\TR School of Computing, Engineering} \& the Built Environment, Edinburgh Napier University, Edinburgh, EH10 5DT, UK.\\
  \AND
 \Name{Matthew C. Turner} \Email{m.c.turner@soton.ac.uk}\\
 \addr School of Electronics and Computer Science, University of Southampton, Southampton, SO17 1BJ, UK.\\
}

\begin{document}

\maketitle

\begin{abstract}%
With neural networks being used to control safety-critical systems, they increasingly have to be both accurate (in the sense of matching inputs to outputs) and robust. However, these two properties are often at odds with each other and a trade-off has to be navigated. To address this issue, this paper proposes a method to generate an approximation of a neural network which is certifiably more robust.  Crucially, the method is fully convex and posed as a semi-definite programme. An application to robustifying model predictive control is used to demonstrate the results. The aim of this work is to introduce a method to
navigate the neural network robustness/accuracy trade-off. 
\end{abstract}


\section{Introduction}
Neural networks have emerged as {\TR powerful} and flexible  nonlinear function approximators for mapping input/output data. For control applications, the main strengths of neural networks are, arguably, their ability to learn complex feedback policies for control problems which are challenging to formulate mathematically. Examples include problems where only data is available (and no model exists) or when the costs/models/constraints may be unknown, as encountered, for instance, in image-based control of autonomous robotics. However, neural networks also have weaknesses. They suffer from a lack of robustness and explainability; two issues which have, so far, prevented their widespread adoption into safety-critical systems. As a consequence, whilst neural networks have thrived as control systems for technologies where failure is not catastrophic (such as robotic demonstrations in the lab), there remains scepticism about their value for systems where crashes can be deadly (such as with aircraft autopilot software). For these safety-critical applications, traditional control theory still dominates. 

The growing appreciation of the need to embed robustness more deeply into neural network design mirrors the concerns raised by the control community during the 1980s in light of high profile incidents, such as the Gripen JAS39 prototype incident detailed in~\cite{stein2003respect}. In the wake of these incidents, robust control theory emerged as a new design paradigm that placed as much emphasis on a control system's robustness as its performance (\cite{green2012linear}), with the control synthesis results building upon earlier robustness analysis results such as the Zames-Falb multipliers of~\cite{zames1968stability}. Arguably, a similar transition has yet to take place with neural networks, even after several  well-publicised incidents, such as deadly autonomous vehicle crashes. The question then arises as to whether the strengths of neural networks (namely, as powerful, yet simple to train, function approximattors) can be balanced against their weaknesses (e.g. their lack of tight robustness guarantees) in a rigorous and quantifiable way.

\textbf{Contribution:} Addressing this balance between neural network robustness and accuracy is the focus of this paper. Specifically, a method to address the following neural network design problem is developed: 
\begin{quote}``\textit{Given a neural network, generate another one which is quantifiably more robust yet has a similar input/output mapping as the  original neural network}.'' \end{quote}
The solution to this problem presented here in Theorem~\ref{thm} is referred to as a  {\TR \em neural network synthesis method}, as it is motivated by the linear matrix inequalities (LMIs) widely used in robust control synthesis~\cite{dullerud2013course}. Crucially, the proposed method is fully convexified in the sense that the weights and biases of the robustified network are obtained from the solutions of a semi-definite programme (SDP) that trade-off robustness against accuracy. This convexification is achieved by: \textit{i}) using the classical slope restrictions on the activation functions and, more crucially, \textit{ii}) using the 1-norm for the robustness bounds instead of the 2-norm commonly used with Lipschitz bound-type results. We demonstrate the validity of our results in an application towards robustifying  model predictive control (MPC) policy. 

\textbf{Literature:} The results of~\cite{pauli2022neural} and~\cite{wang2023direct} are perhaps the most relevant to this paper. Focussing on~\cite{pauli2022neural}, a method to embed robustness (imposed by LMI constraints) into the neural network training was developed. This was achieved by using the ADMM algorithm to switch between training on the data and imposing Lipschitz bounds. A similar approach was developed in~\cite{wang2023direct}, with points of differences being the focus on equilibrium networks and the direct parameterisation of the weights to satisfy Lipschitz bounds. One issue with the approach of~\cite{pauli2022neural} is that, because robustness is measured using the 2-norm (as in Lipschitz bounds), then a bilinearity appears in the training process, as both the neural network weights/biases and the multipliers of the robustness LMIs are decision variables for the solver. This bilinearity is a source of non-convexity for the optimisation problem, and has appeared in earlier synthesis results such as~\cite{drummond2022reduced} on generating reduced-order approximations of large neural networks. By contrast, {\TR here we propose formulating} the robustness bounds in terms of the 1-norm instead of the 2-norm. Moreover, it is proposed to approximate a given neural network which is assumed to give a good model of the data. Theorem~\ref{thm} shows that with these two tweaks to the problem formulation, the trade-off between network robustness and accuracy can be fully convexified and combined within a single SDP- instead of the iterative projection type approach of ADMM. 

There are deep connections between control theory and neural network robustness analysis, most notably through the application of results from absolute stability theory. Early results in this area include~\cite{chu1999bounds}  and~\cite{chu1999stabilization} where the application to neural networks led to the development of a new set of static multipliers for nonlinearities with repeated terms. Similar results were obtained by~\cite{barabanov2002stability} for recurrent neural networks. More recent efforts in this direction include the works of~\cite{wang2022quantitative},~\cite{fazlyab2019efficient} and~\cite{pauli2021linear} where the connection to LMIs and SDP solvers has been made more explicit. One of the major issues when using these SDP formulations is the lack of algorithm scalability  to the large networks often seen in practice. Efforts to address this scalability issue are now under development, for example with~\cite{wang2024scalability} and ~\cite{newton2023sparse} where sparsity of the neural network was exploited. With the connections between neural networks maturing, there is now  a push towards moving away from solving analysis type problems of neural networks and going towards controller synthesis, as exemplified by {\TR papers} such as~\cite{furieri2022neural},~\cite{junnarkar2022synthesis} and~\cite{newton2022stability}.



\subsection*{Notation}
 Real matrices $M$ of size $m \times n$ are denoted $M \in \mathbb{R}^{n \times m}$. The set of positive-definite matrices $M \succ 0$ of size $n$ are  $M \in \mathbb{S}^n_{\succ 0}$. The set of diagonal matrices of dimension $n$ with positive elements are $\mathbb{D}^n_+$. The identity matrix of dimension $n$ is $I_n$.  
 The set of symmetric matrices $A$ of dimension $n$ are $A \in \mathbb{S}^n$. The Hermitian operation is defined such that $He(M) = M+M^\top$. When appropriate, element $(i,j)$ of a matrix $M$ is referred to as $M^{i,j}$. The $\star$-notation for symmetric matrices is used, as in $M = \begin{bmatrix} a & b \\ b  & c\end{bmatrix} = \begin{bmatrix} a & b \\ \star  & c\end{bmatrix}$. Real vectors $p$ of dimension $n$ are denoted $p\ \in \mathbb{R}^n$, non-negative vectors are $p \in \mathbb{R}^n_+$ and non-negative scalars are $p \in \mathbb{R}_+$. The vector of ones of dimension $n$ is $\mathbf{1}_n$ and the vector of zeros of dimension $n$ is ${0}_n$. The $m \times n$ matrix of zeros is $0_{m \times n}$ and that of ones is $1_{m \times n}.$ Nonlinear functions $\phi: \mathbb{R}^n~\to~\mathbb{R}^n$   act component-wise on their  arguments, \textit{i.e.} $\phi(s) = [\phi_1(s_1),\,\ldots,\phi_n(s_n)]^\top$ for $s = [s_1,\,\ldots,s_n]^\top\in \mathbb{R}^{n}$.
\section{Problem setup}
\subsection{The original neural network}
 Consider a neural network $f(u): \mathbb{R}^{n_u} \to \mathbb{R}^{n_{g}}$ in the implicit form of~\cite{el2021implicit}
\begin{subequations} \label{NN1} \begin{align}
x  & = \phi(W_xx + W_uu + b)\label{imp_well_pos}
 \\
 f(u) &= W_{f,x}x+W_{f,u}u+ b_f, 
 \end{align}\end{subequations}
with weights $W_x \in \mathbb{R}^{n \times n}$,  $W_u \in \mathbb{R}^{n \times n_u}$,  $W_{f,x} \in \mathbb{R}^{n_g \times n}$, $W_{f,u} \in \mathbb{R}^{n_g \times n_u}$ and biases  $b \in \mathbb{R}^{n }$,  $b_f \in \mathbb{R}^{n_g}$ and activation functions $\phi(\cdot)$. 

    The primary reasons for using the implicit form of~\cite{el2021implicit} to represent the networks of~\eqref{NN1} in this paper are simply: \textit{i}) it allows a wide class of network architectures (such as recurrent and feed-forward neural networks, but also many others as detailed in~\cite{el2021implicit}) to be represented within a unified framework, \textit{ii}) it makes the notation more compact. However, we expect that, after a minor tweaking of the notation, the presented results could be applied to a wider class of architectures besides those with the structure of~\eqref{NN1}. Finally, it is remarked that  the implicit neural networks of~\eqref{NN1} are structured similarly to others from the literature, most notably deep equilibrium neural networks~\cite{revay2023recurrent, bai2019deep, pokle2022deep}. 

In this work, the neural network of~\eqref{NN1} is regarded as the true mapping of the data, and it is around this mapping that the accuracy of the robustified network is defined. 
%
 Several standard assumptions are imposed on the neural network of~\eqref{NN1}. 
{\flushleft\textbf{Assumption 1}~
The neural network $f(u)$  of~\eqref{NN1} is assumed to be well-posed, in the sense that for every $u \in \mathbb{R}^{n_u}$ there is a unique solution $x$ to the implicit equation \eqref{imp_well_pos}.  }
{\flushleft\textbf{Assumption 2}~
The activation function $\phi(\cdot)$  is diagonal and $[0, 1]$-slope-restricted}. 

 \vspace{2ex}
Recall that a  function $\phi(\cdot): \mathbb{R}^n \to \mathbb{R}^n$ is called diagonal if $(\phi(s))_i = \phi_i(s_i)$ for all $s \in \mathbb{R}^n$ and $i = 1,2\dots, n$. A diagonal function satisfies a $[\delta_1,\delta_2]$-slope restriction, for  given $\delta_1 < \delta_2$, if
\begin{align}
\delta_1 \leq \frac{\phi_i(s_1)-\phi_i(s_2)}{s_1-s_2} \leq \delta_2 \quad  \forall \: s_1, \,s_2 \in \mathbb{R}, \; s_1 \neq s_2\,. \label{def:slop}
\end{align}
The slope restriction assumption for the nonlinearities of neural networks is widely used, e.g. in~\cite{chu1999bounds, barabanov2002stability, fazlyab2019efficient, drummond2022bounding}. {\TR The present $[0,1]$-slope restriction may be relaxed to a $[0,\delta]$-slope restriction for $\delta >0$ by invoking a loop-shifting argument.}


\subsection{The robustified neural network and robustness in the 1-norm}
The goal of this paper is to develop a method to compute a quantifiably robust neural network (in the sense of Definition~\ref{def:robustness} for some $\gamma$, $\gamma_{u,1}$, $\gamma_{u,2}$) which can approximate the input/output map of the original one given by~\eqref{NN1}. To this end, a second neural network is introduced
\begin{subequations} \label{NN2} \begin{align} 
z  & = \phi(\Psi_z z + \Psi_uu + \beta),
 \\
 g(u) &= \Psi_{g,z}z+\Psi_{g,u}u +\beta_g, 
 \end{align}\end{subequations}
with weights $\Psi_z \in \mathbb{R}^{n \times n}$,  $\Psi_u \in \mathbb{R}^{n \times n_u}$,  $\Psi_{g,z} \in \mathbb{R}^{n_g \times n}$,  $\Psi_{g,u} \in \mathbb{R}^{n_g \times n_u}$ and biases  $\beta \in \mathbb{R}^{n }$,  $\beta_g \in \mathbb{R}^{n_g}$. The goal is to compute values for these weights and biases whilst navigating the accuracy/robustness trade-off with respect to the original neural network~\eqref{NN1}. To facilitate the robustness analysis, the following notation will be used to characterise the outputs generated by two inputs $u_1$ and $u_2$ acting on this network. Specifically, we let $(z,u) = (z_i,u_i)$ for $i=1,2$ denote two input/state pairs of~\eqref{NN2}, and set $g_i : = g(u_i)$ as the corresponding outputs. 
%
The following set will be used to define the space of these two inputs. 
\begin{definition}\label{def:u}
For some $\varepsilon_{u,1} \geq 0$, $\varepsilon_{u,2} \geq 0$, two inputs $u_1$ and $u_2$ are constrained by the set $
\mathcal{U}(\varepsilon_1,\varepsilon_2)
: = \{ (u_1,u_2) \: :\:
\|u_1-u_2\|_2^2 \leq \varepsilon_{u,2},~ \|u_1-u_2\|_1 \leq \varepsilon_{u,1}
\}.$
\end{definition}
%
In this work, the following notion of neural network robustness will be used. Compared to earlier work on neural network robustness certificates using SDPs where the 2-norm of the outputs was bounded, such as~\cite{fazlyab2019efficient}, the 1-norm is used {\TR presently}.
\begin{definition}\label{def:robustness}
The neural network~\eqref{NN2} is said to be robust if there exists $\gamma \geq 0$, $\gamma_{u,1} \geq 0$, $\gamma_{u,2} \geq 0$ such that
 \begin{align}\label{robust}
 \|g(u_1)-g(u_2)\|_1 \leq \gamma + \gamma_{u,1}\|u_1-u_2\|_1 +  \gamma_{u,2}\|u_1-u_2\|_2^2 \quad  \forall\: (u_1,\,u_2) \in \mathcal{U}.
 \end{align}
\end{definition}
The following measure of the similarity between two neural networks (\cite{li2021robust}), as in their accuracy with respect to each other's input/output mappings, will be used.
\begin{definition}\label{def:accuracy}
The two neural networks~\eqref{NN1} and~\eqref{NN2} are said to be similar if they have the same biases, as in $\beta = b, \beta_f = b_f$, and their weights are ``close'' in the sense that there exists non-negative scalars  $\varepsilon_{W_x}, \varepsilon_{W_u}, \varepsilon_{W_{f,x}}$, and $ \varepsilon_{W_{f,u}}$ such that, for each matrix element $i,j$,
\begin{equation}\label{norm_wx}
\left\{\begin{aligned}
|W_x^{i,j}-\Psi_z^{i,j}| &\leq \varepsilon_{W_x}, &
\,
|W_u^{i,j}-\Psi_u^{i,j}|&\leq \varepsilon_{W_u}, \\
\,
|W_{f,x}^{i,j}-\Psi_{g,z}^{i,j}| &\leq \varepsilon_{W_{f,x}}, &
\,
|W_{f,u}^{i,j}-\Psi_{g,u}^{i,j}| &\leq \varepsilon_{W_{f,u}}.
\end{aligned}\right.
\end{equation}
If the above holds, then it said that
$   g(u) \approx f(u)$ for all $u \in \mathbb{R}^{n_u}.$
\end{definition}
   Whilst there are limitations with the above definition for network similarity (a more robust version such as that considered in~\cite{li2021robust} may provide improved theoretical nuance), it was observed that this form empirically performs well and, crucially, allows the problem to be convexified. 

\subsection{Problem statement}

Solutions to the following problem are the considered. 

{\flushleft\textbf{Problem 1}}~
 {\TR Determine } weights and biases of the neural network~\eqref{NN2} such that it is robust in the sense of Definition~\ref{def:robustness} 
 and is similar to the original neural network in the sense of Definition~\ref{def:accuracy}. 

\vspace{2ex}
The idea of this problem is to explore the trade-off between neural network accuracy and robustness. The main benefit of the {\TR present results are that they lead to} a convex synthesis problem. 

 

\section{Preliminary results}\label{sec:lemmas}

We present a solution to Problem 1 in Theorem~\ref{thm} below, the proof of which is supported by the preliminary material in this section. 
This result will rely upon the expression of the 1-norm of a vector as a sum of ReLU functions. Recall that
the ReLU function $r(\cdot):  \mathbb{R}^{n_\sigma} \to \mathbb{R}^{n_\sigma}$ is the diagonal function with equal components given by 
\begin{equation}\label{def:relu}
r(\sigma) = \text{ReLU}(\sigma) = \max \big\{0, \sigma \big\} \quad \forall \: \sigma \in \mathbb{R}\,.
\end{equation}
%
This nonlinearity is $[0,1]$-slope-restricted {\TR and} also satisfies a range of other quadratic constraints, as detailed in~\cite{richardson2023strengthened} for example.
The key idea behind the robust neural network synthesis of Theorem~\ref{thm} is the following representation of the absolute value as ReLU functions: 
\begin{align}\label{1-norm}
|\sigma| =\text{ReLU}(\sigma)+ \text{ReLU}(-\sigma)\quad \forall \: \sigma\in \mathbb{R}.
\end{align}

This representation allows the 1-norm of the robustified neural network's output to be expressed as a \textit{linear} sum of [0,1]-slope-restricted nonlinear functions acting on the output. With this linearity, the problem can be convexified. By contrast, using the 2-norm of the Lipschitz bounds introduces a bi-linearity into the matrix inequalities, as also encountered in~\cite{drummond2022reduced}, and so only gives locally optimal results. 



With $r(s) = \text{ReLU}(s)$ as in~\eqref{def:relu}, define the incremental variables
\begin{equation}
\left\{\begin{aligned}
\begin{bmatrix} \tilde{g} \\ \tilde{u} \\ \tilde{z} \end{bmatrix}
&= \begin{bmatrix} g_2-g_1 \\ u_2-u_1 \\ z_2-z_1 \end{bmatrix}
= \begin{bmatrix}\Psi_{g,z}\tilde{z} + \Psi_{g,u}\tilde{u} \\ u_2-u_1 \\ z_2-z_1 \end{bmatrix}
,~
\tilde{g}_\pm = \begin{bmatrix} r( \tilde{g}) \\ r(-\tilde{g}) \end{bmatrix},\\
~
\tilde{u}_\pm  &=  \begin{bmatrix} r(\tilde{u}) \\ r(-\tilde{u}) \end{bmatrix},
\quad \text{and} \quad
p(u_1,u_2)  = \begin{bmatrix} {\tilde{g}_\pm}^\top & {\tilde{u}_\pm}^\top & \tilde{z}^\top & \tilde{u}^\top & 1\end{bmatrix}^\top.
\end{aligned}\right.
\end{equation}
Recall that the inputs $u_1$, $u_2$ and outputs $g_1 = g(u_1)$, $g_2 = g(u_2)$ are given by the neural network~\eqref{NN2}. With these variables, the following lemmas define incremental quadratic constraints for~\eqref{NN2}.

\begin{lemma}\label{lem:sz}
For all $\mathbb{T}_z \in \mathbb{D}_+^n$, {\TR it follows that} 
\begin{align*}
s_z (u_1,u_2) \, {\TR :=} \, p(u_1,u_2) ^\top \Omega_z(\mathbb{T}_z ) p(u_1,u_2)  \geq 0 \quad \forall \: (u_1,\,u_2) \in \mathcal{U}\,,
\end{align*}
with the matrix $\Omega_z(\mathbb{T}_z )$  defined in equation~\eqref{app:lem10} in the Appendix. 
\end{lemma}
\begin{proof}
The quadratic constraint of $s_z (u_1,u_2) $ can be expanded out as
$
\tilde{z}^\top \mathbb{T}_z(\Psi_z\tilde{z}+\Psi_u\tilde{u}-\tilde{z}) \geq 0
$
whose non-negativity follows from the assumed $[0,1]$-slope-restriction of the activation function. 
\end{proof}
\begin{lemma}\label{lem:su}
For all $\mathbb{T}_{u,1} \in \mathbb{R}_+$ and $\mathbb{T}_{u,2} \in \mathbb{R}_+$, {\TR it follows that}
\begin{align*}
s_u(u_1,u_2) \, {\TR :=} \, p(u_1,u_2)^\top {\TR \Omega_u(\mathbb{T}_{u,1},\mathbb{T}_{u,2})} 
p(u_1,u_2) \geq 0 \quad \forall \: (u_1,\,u_2) \in \mathcal{U}\,,
\end{align*}
with the matrices $\Omega_u(\mathbb{T}_{u,1},\mathbb{T}_{u,2})$ defined in equation~\eqref{app:lem11}  in the Appendix. 
\end{lemma}
\begin{proof}
Using the decomposition of the 1-norm in terms of the ReLU of~\eqref{1-norm}, then the quadratic constraint of $s_u(u_1,u_2)$ can be expanded out as
\begin{align*}
p(u_1,u_2) ^\top\Omega_u(\mathbb{T}_{u,1},\mathbb{T}_{u,2})p(u_1,u_2)   & = \mathbb{T}_{u,1} \varepsilon_{u,1} + \mathbb{T}_{u,2} \varepsilon_{u,2} - \frac{1}{2}\mathbb{T}_{u,2} \|\tilde{u}\|_2^2
\\ &\quad    - \frac{1}{2}\mathbb{T}_{u,2} (r(\tilde{u}) ^2+ r(-\tilde{u})^2)  - \mathbb{T}_{u,1} (r(\tilde{u}) + r(-\tilde{u}))  , \nonumber
\\
 & =  \mathbb{T}_{u,1} \varepsilon_{u,1} + \mathbb{T}_{u,2} \varepsilon_{u,2}
 - \mathbb{T}_{u,2} \|\tilde{u}\|_2^2  - \mathbb{T}_{u,1} \|\tilde{u}\|_1  
\geq 0,
\end{align*}
since $(u_1,\,u_2) \in \mathcal{U}$ following Definition~\ref{def:u}.
\end{proof}

\begin{lemma}\label{lem:sg}
For all $\mathbb{T}_g \in \mathbb{D}_+^{n_g}$, {\TR it follows that}
\begin{align*}
s_g(u_1,u_2) \, {\TR :=} \, p(u_1,u_2) ^\top \Omega_g(\mathbb{T}_g ) p(u_1,u_2)  \geq 0 \quad \forall \: (u_1,\,u_2) \in \mathcal{U}\,,
\end{align*}
with the matrix $\Omega_g(\mathbb{T}_{g})$ defined in equation~\eqref{app:lem12} in the Appendix. 
\end{lemma}
\begin{proof}
Similarly to Lemma~\ref{lem:sz}, the quadratic constraint of $s_g(u_1,u_2)$ can be expanded out as \newline
$
r(\tilde{g})^\top~\mathbb{T}_g(\tilde{g}~-~r(\tilde{g}))~+~r(-\tilde{g})^\top~\mathbb{T}_g(-\tilde{g}-r(-\tilde{g}))~\geq~0
$
which is again non-negative {\TR owing to} the [0,1] slope restriction of the ReLU$(\cdot)$ function. 
\end{proof}

\section{Main {\TR result}}

We state the main result of the present work, a theorem containing a SDP to solve Problem 1.

\begin{theorem}\label{thm}
Set the weights $w_0, w_1, w_2 \in \mathbb{R}_{+}$ and the tolerances $\varepsilon_{W_x}, \varepsilon_{W_u}, \varepsilon_{W_{f,x}}, \varepsilon_{W_{f,u}} \in \mathbb{R}_{+}$.  If there exists some $\mathbb{T}_z \in \mathbb{D}^n_{+}$,  $\mathbb{T}_g \in \mathbb{D}^{n_g}_{+}$, $\mathbb{T}_{u_1} \in \mathbb{R}_{+}$,  $\mathbb{T}_{u_2} \in \mathbb{R}_{+}$, $\mathbb{Y}_z \in \mathbb{R}^{n \times n}$, $\mathbb{Y}_u \in \mathbb{R}^{n \times n_u}$, $\mathbb{Y}_{g,z} \in \mathbb{R}^{n_g \times n}$, $\mathbb{Y}_{g,u} \in \mathbb{R}^{n_g \times n_u}$, $\gamma \geq 0$, $\gamma_{u,1} \geq 0$ and $\gamma_{u,2} \geq 0$ that solves
\begin{subequations}\begin{align}
&\min w_0 \gamma+w_1 \gamma_{u_1} + w_2 \gamma_{u,2}, \\
\text{subject to: } & 
\check{\Omega}_z(\cdot) + \check{\Omega}_g(\cdot)+ \Omega_u(\cdot) + \Omega_\gamma(\cdot)  \prec 0, \label{mat_ineq}
\end{align}\end{subequations}
with the matrices $\check{\Omega}_z(\mathbb{T}_z , \mathbb{Y}_z, \mathbb{Y}_u  )$, $ \check{\Omega}_g(\mathbb{T}_{g},\mathbb{Y}_{g,z},\mathbb{Y}_{g,u})$, and $\Omega_\gamma(\gamma, \gamma_{u,1}, \gamma_{u,2})$ defined in~\eqref{eq:thm_matrix_1}--\eqref{eq:thm_matrix_3}, 
and the following \textbf{element-wise} matrix inequalities hold
\begin{subequations}\begin{align}
 - \varepsilon_{W_x} \mathbb{T}_z 1_{n \times n}  & \leq \mathbb{Y}_z-\mathbb{T}_zW \leq \varepsilon_{W_x} \mathbb{T}_z1_{n \times n} , \label{eq:LMI_1}
\\
 - \varepsilon_{W_u} \mathbb{T}_z1_{n \times n_u}  &\leq \mathbb{Y}_u-\mathbb{T}_zW_u \leq \varepsilon_{W_u}  \mathbb{T}_z 1_{n \times n_u} ,
\label{eq:LMI_2} \\
 -\varepsilon_{W_{f,x}} \mathbb{T}_g 1_{n_g \times n} & \leq \mathbb{Y}_{g,z}-\mathbb{T}_gW_{f,x} \leq \varepsilon_{W_{f,x}}  \mathbb{T}_g  1_{n_g \times n} ,
\label{eq:LMI_3} \\
  -\varepsilon_{W_{f,u}} \mathbb{T}_g  1_{n_g \times n_u} &\leq \mathbb{Y}_{g,u}-\mathbb{T}_gW_{f,u} \leq \varepsilon_{W_{f,u}}  \mathbb{T}_g 1_{n_g \times n_u} , \label{eq:LMI_4}
\end{align}\end{subequations}
then, the neural network of $g(u)$ with weights and biases defined by
\begin{align}\label{g_weights}
\Psi_z = {\mathbb{T}_z}^{-1}\mathbb{Y}_z,
\Psi_u = {\mathbb{T}_z}^{-1}\mathbb{Y}_u,
\Psi_{g,z} = {\mathbb{T}_g}^{-1}\mathbb{Y}_{g,z},
\Psi_{g,u} = {\mathbb{T}_g}^{-1}\mathbb{Y}_{g,u},
\beta = b, \beta_f = b_f,
\end{align}
is robust in the sense of Definition~\ref{def:robustness} and similar to $f(u)$ from \eqref{NN1} in the sense of Definition~\ref{def:accuracy}.

\end{theorem}
\begin{proof}
Multiplying the matrix inequality of~\eqref{mat_ineq} on the right by $p(u_1,u_2)$ and on the left by its transpose implies, with the factorisation of the weights in~\eqref{g_weights}, that
\begin{align}
 \|g(u_1)-g(u_2)\|_1 - \gamma - \gamma_{u,1}\|u_1-u_2\|_1 -&  \gamma_{u,2}\|u_1-u_2\|_2^2+s_z(u_1,u_2) \notag \\
 &+s_g(u_1,u_2)+s_u(u_1,u_2) \leq 0\,. \label{g1g2}
\end{align}
Using Lemmas~\ref{lem:sz},~\ref{lem:su} \&~\ref{lem:sg}, then the following quadratic inequalities are non-negative $s_z(u_1,u_2)$, $s_g(u_1,u_2)$, $ s_u(u_1,u_2) \geq 0$ for all $(u_1, u_2) \in \mathcal{U}$. With these conditions, then~\eqref{g1g2} implies that the robustness bound of~\eqref{robust} holds.
To show the norm conditions of~\eqref{norm_wx}, {\TR consider the first set of element-wise inequalities~\eqref{eq:LMI_1}} 
\begin{align*}
- \varepsilon_{W_x} \mathbb{T}_z^{i,i} \leq \mathbb{Y}_z^{i,j}-\mathbb{T}_z^{i,i}W^{i,j} \leq \varepsilon_{W_x} \mathbb{T}_z^{i,i}, \quad \forall \: i,j \in 1, 2, \dots, n
\end{align*}
with superscripts ${i,j}$ denoting matrix element ($i,j$). Since $\mathbb{T}_z \in \mathbb{D}_+^n$, then $\mathbb{T}_z^{i,i} >0$, and so
\begin{align*}
- \varepsilon_{W_x} \leq \Psi^{i,j}-W^{i,j} \leq \varepsilon_{W_x} ,~ \forall i,j \in 1, 2, \dots, n
\end{align*}
giving the norm condition of~\eqref{norm_wx}. The remaining norm bounds of~\eqref{norm_wx} can  be obtained in a similar manner {\TR from~\eqref{eq:LMI_2}--\eqref{eq:LMI_4}}. 
\end{proof}

{\TR Observe from}~\eqref{g_weights} that the weights and biases of the robustified neural network, $\Psi_z $,  $\Psi_u $,  $\Psi_{g,z}  $,  $\Psi_{g,u} $ and biases  $\beta$,  $\beta_g $, can be extracted directly from the matrix variables of the problem:  $\mathbb{T}_z$,  $\mathbb{T}_g $, $\mathbb{T}_{u_1} $,  $\mathbb{T}_{u_2}$, $\mathbb{Y}_z $, $\mathbb{Y}_u $, $\mathbb{Y}_{g,z} $, $\mathbb{Y}_{g,u} $, {\TR $\gamma$}, $\gamma_{u,1} $ and $\gamma_{u,2} $. The conditions of the theorem are therefore linear in the decision variables, and so can be solved as a single SDP. 




\section{Numerical example: Robustifying model predictive control}

To demonstrate the {\TR utility} of Theorem~\ref{thm}  for a control-orientated application, consider the problem of approximating a model predictive control (MPC)  policy to increase its robustness.  {\TR The following example is based on that from~\cite{drummond2022bounding} \footnote{Code to generate these results can be found at: \texttt{https://github.com/r-drummond/convex-nn-synthesis-1norm}.}. 

Consider {\TR the discrete-time, time-invariant, linear control} system
\begin{equation*}
    \begin{bmatrix}
           w_1[k+1]\\  w_2[k+1]
        \end{bmatrix}=
        \begin{bmatrix}
            4/3 & -2/3\\
            1 & 0
        \end{bmatrix}
     \begin{bmatrix}
           w_1[k]\\  w_2[k]
        \end{bmatrix}+
        \begin{bmatrix}
            0 \\ 1
        \end{bmatrix}
    v[k],
\end{equation*}
with state $w[k] \in \mathbb{R}^2$ and {\TR control action $v[k] \in \mathbb{R}$. Define the collection of future states and inputs as $\mathbf{w} = [w[k]^\top,\, w[k+1]^\top \dots,\, w[k+N]^\top ]$ and $\mathbf{v} = [v[k]^\top,\, v[k+1]^\top \dots,\, v[k+N]^\top ]$. The control} action is to be obtained by solving the following quadratic {\TR program:}
\begin{align}\label{QP1}
\text{{\TR minimize:}} &  \quad J(\mathbf{w},\mathbf{v}) =  w[k+N]^{\top}Pw[k+N]\\
& \qquad +\sum_{i=1}^{N-1} w[k+i]^{\top}  Q_iw[k+i]
    +v[k+i-1]^{\top} R_iv[k+i-1], \nonumber
    \\   \text{subject to:} &  \quad G\mathbf{v}  \leq S_vw[k]+c, \notag 
\end{align}
where $Q_i\in \mathbb{S}^{2}_{\succ 0}$, $P\in \mathbb{S}^{2}_{\succ 0}$, $R_i > 0$ $\forall i  = 1,\ldots,  N -1$, and the constraint set is $\mathbf{v} \in \mathcal{V}$. This quadratic program can be more compactly expressed as 
\begin{align*}
\underset{\mathbf{v}}{\mbox{minimize}}~& \mathbf{v}^{\top} H \mathbf{v} + 2w[k]^{\top}F^{\top}\mathbf{v}, \\
    \mbox{subject to:} \quad    & G\mathbf{v}  \leq S_v w[k]+c. \nonumber
\end{align*}
Here, we set the horizon length to $N = 10$ and the control action to be saturated at $-10~\leq~\mathbf{v}~\leq~10$, so $G~=0.1 \times \begin{bmatrix}I_m & -I_m \end{bmatrix}^\top$, $c=[1,\,  1,\, \dots \,1 ,\, 1 ]^\top$ and $S_v = 0$. Furthermore, the MPC quadratic cost function of~\eqref{QP1} is parameterised by
\[
{P} =  \begin{bmatrix}7.1667  & -4.2222\\ -4.2222 & 4.6852\end{bmatrix}, \;  {Q}_i = \begin{bmatrix}1  & -2/3\\ -2/3  & 3/2\end{bmatrix}, \;
{R}_i = 1.
\]
%
%
In a recent paper from~\cite{valmorbida2023quadratic}, it was shown that the MPC control policy can be expressed as
\begin{subequations}\label{mpc_gv}\begin{align}
x&= (I_n-GH^{-1}G^\top)\phi(x) -(S_v+GH^{-1}F)w[k]-c,
\\
v[k] =f(w[k]) &= -H^{-1}Fv[k]-H^{-1}G^\top \phi(x),
\end{align}\end{subequations}
with $\phi(x)$ being the ReLU function. By defining
$\hat{x} =  (I_n-GH^{-1}G^\top)^{-1}(x+Sw[k]+c) = \phi(x),$
it was shown in  \cite{drummond_l4dc} that the MPC policy of~\eqref{mpc_gv} can be re-written as 
\begin{subequations}\begin{align*}
\hat{x} &= \phi( (I_n-GH^{-1}G^\top)\hat{x}-(S_v+GH^{-1}F)w[k]-c) ,
\\
v[k] = g(w[k])  & = -H^{-1}Fw[k]-H^{-1}G^\top \phi(x)\\
&= -H^{-1}Fw[k]-H^{-1}G^\top \hat{x}\,,
\end{align*}\end{subequations}
{\TR that is,} as an implicit neural network in the form of~\eqref{NN1} with  weights and biases
\begin{align*}
W_x &= I_n-GH^{-1}G^\top, \; W_u = (S_v+GH^{-1}F), \; b = c, \\
W_f &= -H^{-1}G^\top, \; W_{f,u} = -H^{-1}F.
\end{align*}


 Theorem~\ref{thm} can then be directly applied to this MPC problem. A neural network of the form of~\eqref{NN2} which trades off accuracy relative to the original MPC policy and its own robustness measured by Definition~\ref{def:robustness} can then be synthesized.  In this example, the tolerances for the weights and biases are set to a fixed value of $ \varepsilon_{W_x} = \varepsilon_{W_u} = \varepsilon_{W_{f,x}} = \varepsilon_{W_{f,u}} =  \varepsilon$. Increasing $\varepsilon$ opens up the search space for the weights,  and so increases the neural network's robustness as it is that property which is optimised by Theorem~\ref{thm}. But, it also means that the gap between the weights of the original and robustified networks can be increased, which can reduce the accuracy relative to the original network of~\eqref{NN1}. Tuning this tolerance $\varepsilon$ allows the trade-off between accuracy and robustness of~\eqref{NN2} to be navigated. 
 \begin{figure}[t]
         \centering
         \includegraphics[width=0.95\textwidth]{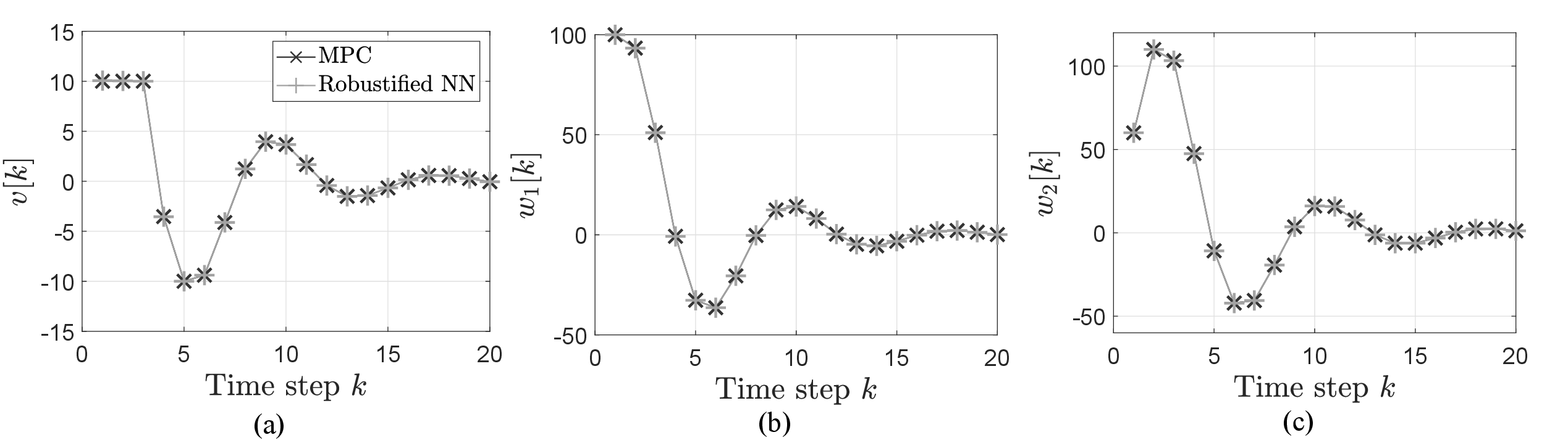}
         \caption{Comparison between MPC controller and neural network generated by Theorem~\ref{thm} with a tolerance of $\varepsilon = 10^{-5}$. (a) Compares the inputs $v[k]$, (b) the state $w_1[k]$, and (c) the state $w_2[k]$.}
         \label{fig:ex1}
     \end{figure}

 \begin{figure}[t]
         \centering
         \includegraphics[width=0.95\textwidth]{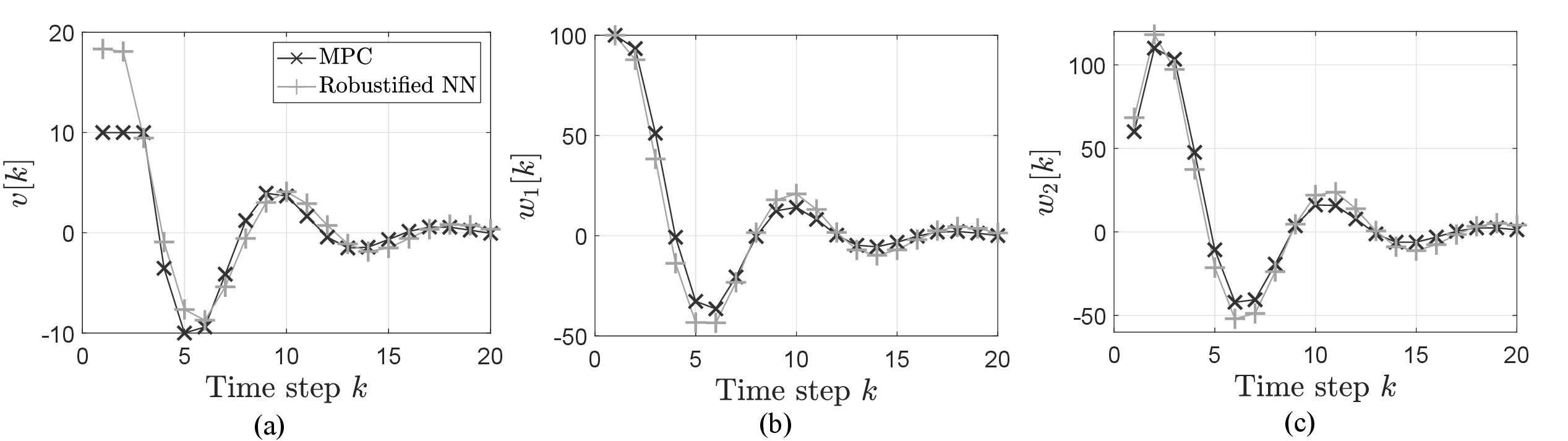}
         \caption{Comparison between MPC controller and neural network generated by Theorem~\ref{thm} with a tolerance of $\varepsilon = 10^{-1}$. (a) Compares the inputs $v[k]$, (b) the state $w_1[k]$, and (c) the state $w_2[k]$.}
         \label{fig:ex2}
     \end{figure}

Figures~\ref{fig:ex1} and~\ref{fig:ex2} compare the response of both the original MPC problem and the robustified neural network generated by Theorem~\ref{thm}. Figures~\ref{fig:ex1} prioritises accuracy with respect to the MPC policy, with $\varepsilon = 10^{-5}$ whereas Figure~\ref{fig:ex2} prioritises robustness, $\varepsilon = 10^{-1}$. This trade-off is captured by the responses, with the $\varepsilon = 10^{-5}$ response able to replicate that of the MPC (and capture the input saturation limits) unlike the $\varepsilon = 10^{-1}$ response. 

 \begin{wrapfigure}{r}{0.45\textwidth}
  \begin{center}
    \includegraphics[width=0.45\textwidth]{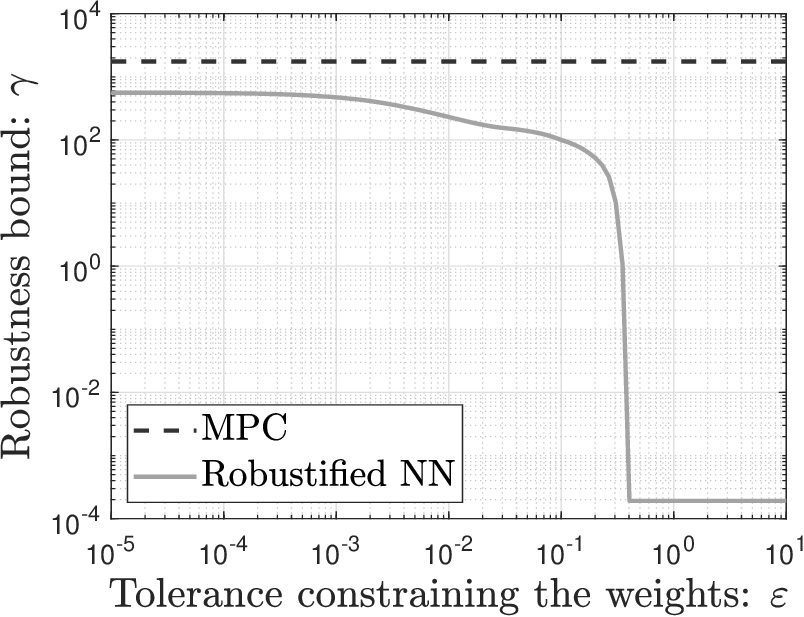}
  \end{center}
  \caption{Trade-off between robustness (measured by $\gamma$ of Definition \ref{def:robustness}) and accuracy (measured by $\varepsilon$ of Definition \ref{def:accuracy}) for the MPC problem.} 
  \label{fig:robustness}
\end{wrapfigure}

A reduction in robustness is the price payed for this increased accuracy with respect to the original MPC policy. Figure~\ref{fig:robustness} demonstrates this trade-off, which plots the robustified neural networks 1-norm bound against the weight tolerance $\varepsilon$. Using the definition of Definition~\ref{def:robustness}, the gains $\gamma_{u,1} = \gamma_{u,2} = 0$ were set to have a single valuable capturing the bound for the 1-norm of the outputs. As  $\varepsilon$ increased, the bound $\gamma$ reduced-- indicating that a more robust network was generated. At $\varepsilon \approx 0.3$, the bound dropped suddenly because the solver had found the most robust network, which in this case was to map the states to zero. Even though this is not a particularly useful network, it is ``optimally robust'', a feature which emphasises the need to trade-off robustness against accuracy. Interestingly, even when $\varepsilon$ was small and so could accurately capture the MPC response (see Figure~\ref{fig:ex1}), the synthesized network was more robust, with $\gamma 
 = 1,751$ of the MPC reduced to $\gamma = 560$ for the neural network.

 





\newpage
\section*{Conclusions}
A method to approximate a neural network by a more robust version was developed. By defining the robustness problem in terms of the 1-norm of the output, the proposed robust neural network synthesis method was shown to be a convex semi-definite programme. An application to robustifying model predictive control feedback policies demonstrated the applicability of the approach.


\appendix
\section{}

{\TR We provide the matrices not given in the main text.} The matrices of Lemmas~\ref{lem:sz}, \ref{lem:su} and~\ref{lem:sg} are, respectively,
{\footnotesize\begin{align}\label{app:lem10}
 \Omega_z(\mathbb{T}_z )  = 
\begin{bmatrix} 0_{2n_g \times 2n_g} & 0_{2n_g \times 2n_u}  & 0_{2n_g \times n} &  0_{2n_g \times n_u}&  0_{2n_g } \\ \star &0_{2n_u \times 2n_u}  & 0_{2n_u \times n} &  0_{2n_u \times n_u} & 0_{2n_u}  \\\ \star  &\star  & He( \mathbb{T}_z \Psi_z-  \mathbb{T}_z) &  \mathbb{T}_z \Psi_u  & 0_{n}  \\ \star  &\star & \star  &   0_{n_u \times n_u} & 0_{n_u }\\  \star  & \star  &  \star  & \star  & 0 \end{bmatrix}\,,
\end{align}}
%
{\footnotesize
\begin{align}\label{app:lem11}
 \Omega_u(\mathbb{T}_{u,1},\mathbb{T}_{u,2})  = 
\begin{bmatrix} 0_{2n_g \times 2n_g} & 0_{2n_g \times 2n_u}  & 0_{2n_g \times n} &  0_{2n_g \times n_u}&  0_{2n_g } \\ 
\star & - \frac{1}{2} \mathbb{T}_{u,2}\begin{bmatrix}I_{n_u}  & 0_{n_u \times n_u} \\ \star  &  I_{n_u} \end{bmatrix}  & 0_{2n_u \times n} &  0_{2n_u \times n_u} & -\frac{1}{2}\mathbb{T}_{u,1}  \begin{bmatrix} \mathbf{1}_{n_u}  \\ \mathbf{1}_{n_u}   \end{bmatrix}
\\ \star  &\star  & 0_{n \times n} & 0_{n \times n_u} & 0_{n}  
\\ \star  &\star & \star  &   -\frac{1}{2}\mathbb{T}_{u,2} I_{n_u}& 0_{n_u }
\\  \star  & \star  &  \star  & \star  &  \mathbb{T}_{u,1} \varepsilon_{u,1}+\varepsilon_{u,2} \mathbb{T}_{u,2}\end{bmatrix}\,,
\end{align}
}
%
{\footnotesize
\begin{align}\label{app:lem12}
 \Omega_g(\mathbb{T}_{g})  = 
\begin{bmatrix}\begin{bmatrix}-\mathbb{T}_{g} & 0_{n_g \times n_g} \\ \star  &-\mathbb{T}_{g} \end{bmatrix}  & 0_{2n_g \times 2n_u}  &  \begin{bmatrix}\mathbb{T}_g \Psi_{g,z} \\ - \mathbb{T}_g \Psi_{g,z}  \end{bmatrix} &    \begin{bmatrix}\mathbb{T}_g\Psi_{g,u} \\ - \mathbb{T}_g\Psi_{g,u}  \end{bmatrix}&  0_{2n_g } \\ 
\star & 0_{2n_u \times 2n_u} & 0_{2n_u \times n} &  0_{2n_u \times n_u} & 0_{2n_u}  
\\ \star  &\star  & 0_{n \times n} & 0_{n \times n_u} & 0_{n}  
\\ \star  &\star & \star  & 0_{n_u \times n_u} & 0_{n_u }
\\  \star  & \star  &  \star  & \star  & 0\end{bmatrix}\,.
\end{align}
}
Finally, the following matrices are used in Theorem~\ref{thm}.
{\footnotesize
\begin{align}
\check{\Omega}_z(\mathbb{T}_z , \mathbb{Y}_z, \mathbb{Y}_u  )  & = 
\begin{bmatrix} 0_{2n_g \times 2n_g} & 0_{2n_g \times 2n_u}  & 0_{2n_g \times n} &  0_{2n_g \times n_u}&  0_{2n_g } \\ \star &0_{2n_u \times 2n_u}  & 0_{2n_u \times n} &  0_{2n_u \times n_u} & 0_{2n_u}  \\\ \star  &\star  & He(\mathbb{Y}_z-  \mathbb{T}_z) & \mathbb{Y}_u  & 0_{n}  \\ \star  &\star & \star  &   0_{n_u \times n_u} & 0_{n_u }\\  \star  & \star  &  \star  & \star  & 0 \end{bmatrix}, \label{eq:thm_matrix_1}
\\
 \check{\Omega}_g(\mathbb{T}_{g},\mathbb{Y}_{g,z},\mathbb{Y}_{g,u})   & = 
\begin{bmatrix}\begin{bmatrix}-\mathbb{T}_{g} & 0_{n_g \times n_g} \\ \star  &-\mathbb{T}_{g} \end{bmatrix}  & 0_{2n_g  \times 2n_u}  &  \begin{bmatrix}\mathbb{Y}_{g,z} \\ - \mathbb{Y}_{g,z}  \end{bmatrix} &   \begin{bmatrix} \mathbb{Y}_{g,u} \\ - \mathbb{Y}_{g,u}  \end{bmatrix}&  0_{2n_g  } \\ 
\star & 0_{2n_u \times 2n_u} & 0_{2n_u \times n} &  0_{2n_u \times n_u} & 0_{2n_u}  
\\ \star  &\star  & 0_{n \times n} & 0_{n \times n_u} & 0_{n}  
\\ \star  &\star & \star  & 0_{n_u \times n_u} & 0_{n_u }
\\  \star  & \star  &  \star  & \star  & 0\end{bmatrix}, \label{eq:thm_matrix_2}
\\
 \Omega_\gamma(\gamma, \gamma_{u,1}, \gamma_{u,2})   & = 
\begin{bmatrix} 0_{2n_g \times 2n_g} & 0_{2n_g \times 2n_u}  & 0_{2n_g \times n} & 0_{2n_g \times n_u}& \frac{1}{2} \begin{bmatrix} \mathbf{1}_{n_g}  \\ \mathbf{1}_{n_g}   \end{bmatrix} \\ 
\star & - \frac{1}{2}\gamma_{u,2}\,I_{2n_u}  & 0_{2n_u \times n} &  0_{2n_u \times n_u} & -\frac{\gamma_{u,1}}{2}  \begin{bmatrix} \mathbf{1}_{n_u}  \\ \mathbf{1}_{n_u}   \end{bmatrix}
\\ \star  &\star  & 0_{n \times n} & 0_{n \times n_u} & 0_{n}  
\\ \star  &\star & \star  &   -\frac{1}{2}\gamma_{u,2}I_{n_u} & 0_{n_u }
\\  \star  & \star  &  \star  & \star  & -\gamma\end{bmatrix}\,. \label{eq:thm_matrix_3}
\end{align}
}

\acks{Ross Drummond was supported by a UK Intelligence Community Research Fellowship from the Royal Academy of Engineering. }

\bibliography{bibliog}
\end{document}